\documentclass[runningheads]{llncs}

\usepackage{amssymb}
\usepackage{amsmath}
\usepackage{mathtools}
\usepackage{relsize}

\usepackage{framed}
\usepackage{url}
\usepackage{algpseudocode}
\usepackage{algorithm}
\usepackage{cite}

\usepackage{xcolor}
\usepackage{comment}

\usepackage{enumitem}
\usepackage{xcolor}
\usepackage{comment}
\usepackage{tikz}
\usepackage{epsfig}
\usepackage{wrapfig}
\usetikzlibrary{shapes,backgrounds, patterns}

\usepackage{xcolor}
\usepackage{comment}
\usepackage{tcolorbox}
\tcbuselibrary{breakable}

\colorlet{bscolor}{blue}
\colorlet{vrcolor}{red}

\newcommand{\Omit}[1]{}

\usepackage{hyperref,lipsum}

\hypersetup{
  colorlinks   = true, 
  urlcolor     = black, 
  linkcolor    = blue, 
  citecolor   = red 
}
\usepackage{array}

\newcommand{\scp}{\textsc{Star Coloring}}

\usepackage{xcolor}
\usepackage{comment}
\usepackage{tcolorbox}
\tcbuselibrary{breakable}
\newtcolorbox{mybox}[2][]{colbacktitle=white,colback=white,coltitle=black,title={#2},fonttitle=\bfseries,#1, left = 1mm, right = 2mm, breakable}
\usepackage{todonotes}

\newtheorem{reduction rule}{Reduction Rule}

\newtheorem{claim claim}{Claim}

\title{On Structural Parameterizations of Star Coloring}
\titlerunning{}

\author{Sriram Bhyravarapu\inst{1} \and
I. Vinod Reddy\inst{2} 
}
\authorrunning{S. Bhyravarapu and I. Vinod Reddy}
\institute{The Institute of Mathematical Sciences, HBNI, Chennai, India \and
Department of Electrical Engineering and Computer Science,  IIT Bhilai, India
\\
\email{sriramb@imsc.res.in, vinod@iitbhilai.ac.in}}

\begin{document}

\maketitle

\begin{abstract}
A \emph{star coloring} of a graph $G$ is a proper vertex coloring such that
every path on four vertices uses at least three distinct colors. 
The minimum number of colors required for such a star coloring of $G$ is called star chromatic number, denoted by $\chi_s(G)$. 
Given a graph $G$ and a positive integer $k$, the \textsc{Star Coloring Problem} asks whether $G$ has a star coloring using at most $k$ colors. 
This problem is {\sf NP}-complete even on restricted graph classes such as bipartite graphs. 

In this paper, we initiate a study of \scp{} from the parameterized complexity perspective. 
We show that \scp{} is fixed-parameter tractable when parameterized by (a) neighborhood diversity, (b) twin-cover,  
and (c) the combined parameters clique-width and the number of colors.

\end{abstract}
\section{Introduction}
A coloring $f:V(G) \rightarrow \{1,2,\ldots,k\}$ of a graph $G=(V,E)$ is a \emph{star coloring} if (i) $f(u)\neq f(v)$ for every edge $uv\in E(G)$, and (ii) every path on four vertices uses at least three distinct colors. 
The \emph{star chromatic number} of $G$, denoted by $\chi_s{(G)}$, is the smallest integer $k$ such that $G$ is star colorable using $k$ colors. Given a graph $G$ and a positive integer $k$, the \scp{} problem asks whether $G$ has a star coloring using at most $k$ colors.
The name star coloring is due to the fact that the subgraph induced by any two color classes (subset of vertices assigned the same color) is a disjoint union of stars.

\scp{}~\cite{gebremedhin2009efficient} is used in the computation of the Hessian matrix. 
A Hessian matrix is a square matrix of second order partial derivatives of a scalar-valued function.  Hessian matrices are used in large-scale optimization problems, parametric sensitivity analysis~\cite{buskens2001sensitivity}, image processing, computer vision~\cite{lorenz1997multi}, and control of dynamical systems in real time~\cite{buskens2001sensitivity}. Typically, Hessian matrices that arise in a large-scale application are sparse. The computation of a sparse Hessian matrix using the automatic differentiation technique requires a seed matrix. Coleman and Moré~\cite{coleman1984estimation} showed that the computation of a seed matrix  can be formulated using a star coloring of the adjacency graph of a Hessian matrix.


\scp{} was first introduced by Gr{\"u}nbaum in~\cite{grunbaum1973acyclic}. The computational complexity of the problem is studied on several graph classes. The problem is polynomial time solvable on cographs~\cite{lyons2011acyclic}  and line graphs of trees~\cite{omoomi2018polynomial}. 
It is {\sf NP}-complete to decide if there exists a star coloring of bipartite graphs~\cite{coleman1983estimation} using at most $k$ colors, for any $k\geq 3$.
It has also been shown that \scp{}  is {\sf NP}-complete on planar bipartite graphs~\cite{albertson2004coloring} and line graphs of subcubic graphs~\cite{lei2018star} when $k=3$.  Recently, Shalu and Cyriac~\cite{shalu2022complexity} showed that $k$-\textsc{Star Coloring} is {\sf NP}-complete 
for graphs of degree at most four, where $k\in \{4, 5\}$. 


To the best of our knowledge, the problem 
has not been studied in the framework of parameterized complexity.  
In this paper, we initiate the study of \scp{} from the viewpoint of parameterized complexity. In parameterized complexity, the running time of an algorithm is measured as a function of input and a secondary measure called a parameter. A parameterized problem is said to be fixed-parameter tractable (FPT) with respect to a parameter $k$, if the problem can be solved in $f(k) n^{O(1)}$ time, where $f$ is a computable function independent of the input size $n$ and $k$ is a parameter associated with the input instance.  For more details on parameterized complexity, we refer the reader to the texts~\cite{cygan2015parameterized}. 
As \scp{} is {\sf NP}-complete even when $k=3$, 
the problem is para-{\sf NP} complete when parameterized by the number colors $k$. This motivates us to study the problem with respect to structural graph parameters, which measure the structural properties of the input graph. 
The parameter tree-width~\cite{robertson1983graph} introduced by Robertson and Seymour is one of the most investigated structural graph parameters for graph problems.

The \scp{} problem is expressible in monadic second order logic (MSO)~\cite{harshita2017fo}. 
Using the meta theorem of Courcelle~\cite{courcelle1992monadic}, 
the problem is FPT when parameterized by the tree-width  of the input graph. Clique-width~\cite{courcellecw} is another graph parameter which is a generalization of tree-width. 
If a graph has bounded tree-width, then it has  bounded clique-width, however, the converse may not always be true
(e.g., complete graphs). Courcelle's meta theorem can also be extended to graphs of bounded clique-width. It was shown in~\cite{courcelle2000linear} that all problems expressible in MSO logic that does not use edge set quantifications (called as $MS_1$-logic) are FPT when parameterized by the clique-width. However, the \scp{} problem cannot be expressed in $MS_1$ logic~\cite{harshita2017fo,fomin2010intractability}. 
Motivated by this, we study the parameterized complexity of the problem with respect to the combined parameters clique-width and the number of colors and show that \scp{} is FPT.




Next, we consider the parameters neighborhood diversity~\cite{lampis2012algorithmic}  and twin-cover~\cite{ganian2015improving}. 
These parameters are weaker than clique-width in the sense that graphs of bounded neighborhood diversity (resp. twin-cover) have bounded clique-width, however, the converse may not always be true. Moreover, these two parameters are not comparable with the parameter tree-width and they generalize the parameter vertex cover~\cite{ganian2015improving} (see Fig~\ref{fig:my_label-par}). We show that \scp{} is FPT with respect to neighborhood diversity or twin-cover.

\begin{figure}
    \centering
    \includegraphics [trim=2.7cm 17.7cm 1cm 3.5cm, clip=true, scale=0.7]{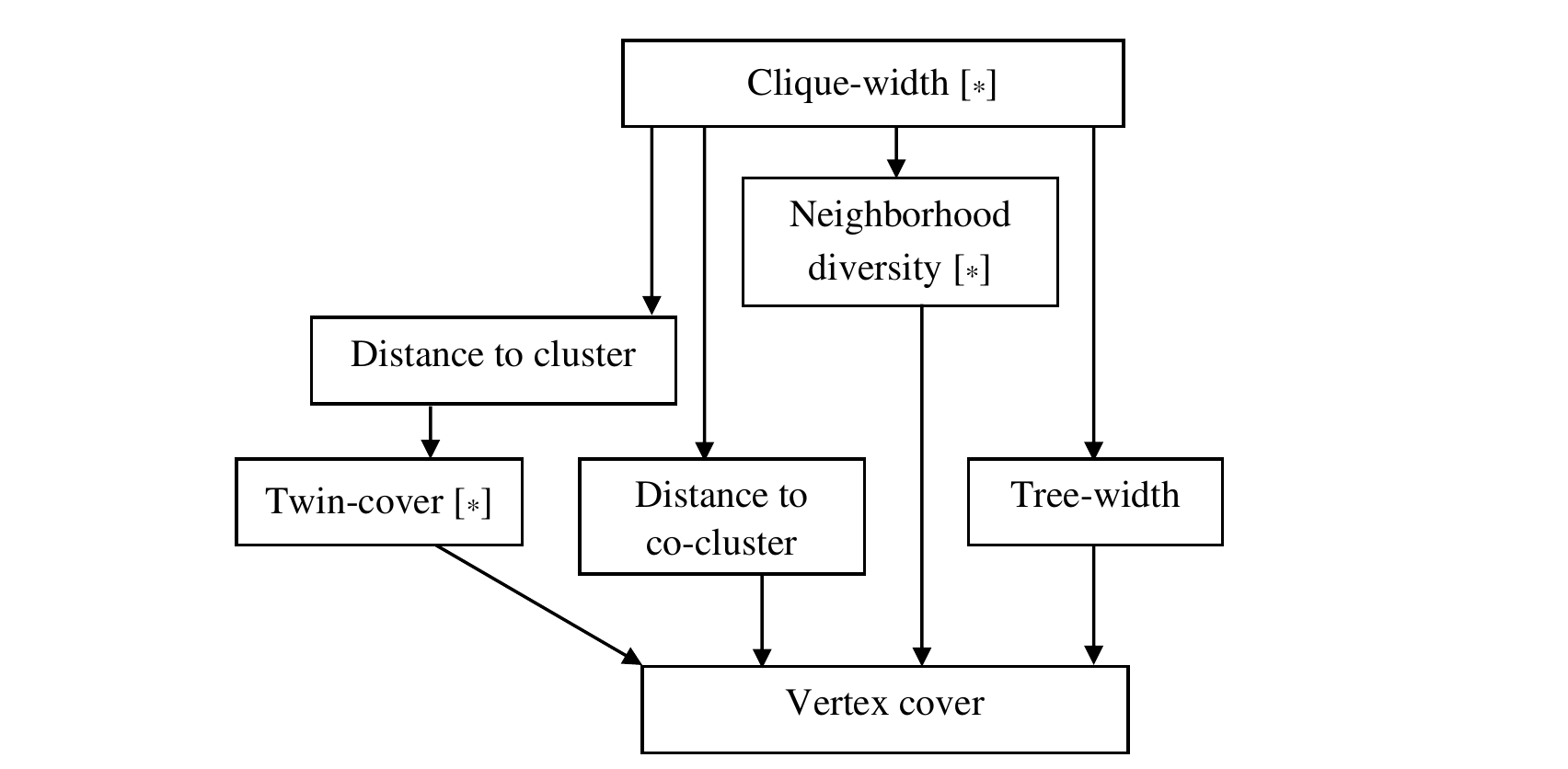}
    \caption{Hasse diagram of some  structural graph parameters. An edge from a parameter $k_1$ to a parameter $k_2$  means that there is a function $f$  such that for all graphs $G$, we have $k_1(G) \leq f(k_2(G))$. The parameters considered in this paper are indicated by $\ast$. }
    \label{fig:my_label-par}
\end{figure}
 
\section{Preliminaries}
For $k \in \mathbb{N}$, we use $[k]$ to denote the set $\{1,2,\ldots,k\}$. If $f: A \rightarrow B$ is a function and $C \subseteq A$, $f|_C$ denotes the restriction of $f$ to $C$, that is $f|_C: C \rightarrow B$ such that for all $x \in C$, $f|_C(x)=f(x)$. 
All graphs we consider in this paper are undirected, connected, finite and simple. For a  graph $G=(V,E)$, we denote the vertex set and edge set of $G$ by $V(G)$ and $E(G)$ respectively. We use $n$ to denote the number of vertices and $m$ to denote the number of edges of the graph.  
For simplicity, an edge between vertices $x$ and $y$ is denoted as $xy$.  
For a  subset $X \subseteq V(G)$, the graph $G[X]$ denotes the subgraph of $G$ induced by vertices of $X$.  If $f:V(G) \rightarrow[k]$ is a coloring of $G$ using $k$ colors, then we use $f^{-1}(i)$ to denote the subset of vertices of $G$ which are assigned the color $i$. For a subset $U \subseteq V(G)$, we use $f(U)$ to denote the set of colors used to color the vertices of $U$, i.e., $f(U)=\bigcup \limits_{u \in U} f(u)$.

For a vertex set $X \subseteq V(G)$, we denote $G - X$, the graph obtained from $G$ by deleting all vertices of $X$ and their incident edges. 
The open neighborhood of a vertex $v$, denoted $N(v)$, is the set of vertices adjacent to $v$ and the set $N[v]=N(v) \cup \{v\}$ denotes  the closed neighborhood of $v$. 
The neighbourhood of a vertex set $S \subseteq V(G)$ is $N(S)=(\cup_{v \in S} N(v)) \setminus S$. For a fixed coloring of $G$, 
we say a path is \emph{bi-colored} if there exists a proper coloring of the path using two colors. 



\section{Neighborhood Diversity}\label{sec:nd}
In this section, we show that \scp{} is 
FPT when parameterized by neighborhood diversity. 
The key idea is to reduce star coloring on graphs of bounded neighborhood diversity 
to the integer linear programming problem (ILP). The latter is FPT when parameterized by the number of variables. 

\begin{theorem}[\cite{frankilp,kannanilp,lenstrailp}]\label{thm:ndilp}
The $q$-variable \textsc{Integer Linear Programming Feasibility} problem can be solved using $O(q^{2.5q+o(q)}n)$
arithmetic operations and space polynomial in $n$, where $n$ is the number of bits of the input.
\end{theorem}

We now define the parameter neighborhood diversity and state some of its properties. 
\begin{definition}[Neighborhood Diversity~\cite{lampis2012algorithmic}]\label{def:nd}
Let $G=(V,E)$ be a graph. 
Two vertices $u,v\in V(G)$ are said to have the \emph{same type} if and only if $N(u)\setminus \{v\}=N(v)\setminus \{u\}$. A graph $G$ has neighborhood diversity at most $t$ 
if there exists a  partition of $V(G)$ into at most $t$ sets  
$V_1, V_2, \dots, V_t$ 
such that all vertices in each set have same type. 
\end{definition}
Observe that each $V_i$   either forms a  clique or an independent set in $G$.
Also, for any two distinct types
$V_i$ and $ V_j$, 
either each vertex in $V_i$ is adjacent to each vertex in $V_j$,
or no vertex in $V_i$ is adjacent to any vertex in $V_j$. 
We call a set $V_i$ as a \emph{clique type} (resp, independent type) if 
$G[V_i]$ is a clique (resp, independent set). It is known that a smallest sized partition of $V(G)$ into clique types and independent types can be found in polynomial time \cite{lampis2012algorithmic}. Hence, we assume that the types $V_1, V_2, \dots, V_t$ of the graph $G$ are given as input. 

We now present the main result of the section. 
\begin{theorem}\label{thm:nd}
\scp{} can be solved in $O(q^{2.5q+o(q)}  n)$ time, where $q=2^t$ and $t$ is the neighborhood diversity of the graph.  
\end{theorem}




 Let $G=(V,E)$ be a graph with the types $V_1, V_2, \dots, V_t$. 
 For each $A\subseteq \{1, 2, \dots, t\}$, we denote a \emph{subset type} of $G$ by $T_A=\{V_i\mid i\in A\}$. 
We denote the set of all types adjacent to type $V_i$ by $adj(V_i)$. That is, $V_j\in adj(V_i)$ if every vertex in $V_j$ is adjacent to every vertex of $V_i$. 
Given a graph $G$ and its types, we construct the ILP instance in the following manner. 


\medskip
\noindent \textbf{Construction of the ILP instance:} 
For each $A\subseteq [t]$, 
let $n_A$ be the variable that denotes 
the number of colors assigned to vertices in every type of 
$T_A$ and not used in any of the types from $\{V_1, V_2, \dots, V_t\}\setminus T_A$. 
For example, if $A = \{1, 3, 4\}$ (i.e., $T_A = \{V_1, V_3 , V_4 \}$) and $n_A = 2$, then there are two colors, say
$c_1$ and $c_2$, such that both $c_1$ and $c_2$ 
are assigned to vertices in 
each of the types $V_1$, $V_3$ and $V_4$ and 
not assigned to any
of the vertices in types $\{V_1 , V_2 , \dots , V_t\} \setminus \{V_1, V_3 , V_4 \}$. This is the critical part of the proof where we look at how many colors are used exclusively in each type of $T_A$ rather than what colors are used. 
Since we have a variable $n_A$ for each $A\subseteq [t]$, 
the number of variables is $2^t$. 


We now describe the constraints for ILP with a short description explaining the significance or the information being captured by the constraints. 
\begin{enumerate}\setlength\itemsep{1.2em}


    \item[(C0)] Discard all subset types $T_A$ containing two types $V_i, V_j$ where  $V_j\in adj(V_i)$. 
    

     To ensure that no two adjacent vertices 
    are assigned the same color, we introduce this constraint that 
    only considers $T_A$ in which no two types in $T_A$ are adjacent.


    

    \item[(C1)] The sum of all the variables is at most $k$. That is $\sum\limits_{A\subseteq[t]}n_A\leq k$. 
    

    We introduce this constraint to ensure that  the number of colors used in any coloring is at most $k$.   
    
    \item[(C2)]

    For each clique type $V_i$, $i\in [t]$, 
    the sum of the variables $n_A$ 
    for which $V_i\in T_A$ is equal to the number of vertices in  $V_i$. 
    That is, $\sum\limits_{A:V_i\in T_A}n_A=|V_i|$.  
 
    To ensure that no two vertices in the clique type $V_i$ are assigned the same color, we introduce this constraint. 
    

    \item[(C3)] For each independent type $V_i$, where $i\in [t]$, 
    the sum of the variables $n_A$ for 
    which $V_i\in T_A$ is at least one and 
    at most the minimum of $k$ and the number of vertices in  $V_i$. 
    That is, 
    $1\leq \sum\limits_{A:V_i\in T_A} n_A \leq  \min\{k, |V_i|\}$. 

    To ensure that the number of colors used for coloring an independent type $V_i$ is at least one and at most the minimum of $k$ and $|V_i|$, we introduce this constraint. 
    

    \item[(C4)] For each combination of four distinct types, 
    say $V_{i_1}, V_{i_2}, V_{i_3}$ and  $V_{i_4}$, where $i_1, i_2, i_3, i_4\in [t]$, with $V_{i_1}, V_{i_3}\in adj(V_{i_2})$ and $V_{i_4}\in adj(V_{i_3})$, we have the following constraint:

        If the sum of the variables $n_A$ for which $V_{i_1}, V_{i_3}\in T_A$ 
    is at least one, then 
    sum of variables $n_{B}$ for which $V_{i_2}, V_{i_4}\in T_{B}$ should be equal to zero. 
    That is, 
$$\sum\limits_{\substack{A:V_{i_1}, 
    V_{i_3}\in  T_A \mbox{ where } \\ V_{i_1}, V_{i_3}\in adj(V_{i_2})  
    \mbox{ and }  V_{i_4}\in adj(V_{i_3})}}
    n_A \geq 1 \implies \sum\limits_{B: V_{i_2}, V_{i_4}\in T_{B}} n_{B} =0.$$

  This constraint ensures that if there exists a vertex in $V_{i_1}$ and a vertex in $V_{i_3}$ that are assigned the same color, then the sets of colors used to color the vertices of $V_{i_2}$ and $V_{i_4}$ are disjoint.

    \item[(C5)] For every combination of three distinct types, say  $V_{i_1}, V_{i_2}, V_{i_3}$, where $i_1, i_2, i_3\in [t]$, with 
    $V_{i_1}$ being an independent type and $V_{i_2}, V_{i_3}\in adj(V_{i_1})$, we have the following constraint: 
    
    If the sum of 
    the variables $n_A$ for which $V_{i_1}\in T_A$ is strictly less than the number of vertices in $V_{i_1}$, then the sum of variables $n_{B}$ for which $V_{i_2}, V_{i_3}\in T_{B}$ is equal to zero.

          $$\sum\limits_{\substack{A:V_{i_1}\in T_A, \mbox{ where } V_{i_2}, V_{i_3}\in adj(V_{i_1}) \mbox{ and } \\ V_{i_1} \mbox{ is an independent type}}} n_A < |V_{i_1}| \implies \sum\limits_{B: V_{i_2}, V_{i_3}\in T_B 
    } n_B =0.$$

    This constraint ensures that if there exist two vertices in $V_{i_1}$ that are assigned the same color, then every vertex in $V_{i_2}$ is assigned a color different from every vertex in $V_{i_3}$. 


    

    


    \item[(C6)]  For every combination of two distinct independent types 
    $V_{i_1}, V_{i_2}$, where $i_1, i_2\in [t]$ with $V_{i_1}\in adj(V_{i_2})$, if the sum of 
    the variables $n_A$ for which $V_{i_1}\in T_A$ is less than the number of vertices in $V_{i_1}$, then the sum of variables $n_{B}$ for which $V_{i_2}\in T_{B}$ is equal to the number of vertices in $V_{i_2}$, and vice-versa. The former constraint is illustrated below while the latter constraint can be constructed by swapping $V_{i_1}$ and $V_{i_2}$ in the former constraint. 
    
     $$\sum\limits_{\substack{A:  V_{i_1}\in  T_A \mbox{ where } V_{i_1}\in adj(V_{i_2}) \\ \mbox{ and } V_{i_1}, V_{i_2} \mbox{ are independent types} }} n_A < |V_{i_1}| \implies \sum\limits_{\substack{B: V_{i_2}\in T_B }} n_B = |V_{i_2}|.$$

     This constraint ensures that if there exist two vertices in $V_{i_1}$ that are assigned the same color  then all vertices in 
     $V_{i_2}$ 
     are assigned distinct colors. We can say similar things for the latter constraint. 
  \item[(C7)] For each $A\subseteq [t]$, $n_A\geq 0$. 
     
     The number of colors used exclusively in all the types of $T_A$ is at least 0.

         \end{enumerate}
        
    The construction of the ILP instance is complete. 
    We use Theorem \ref{thm:ndilp} to obtain a feasible assignment for ILP. Using this, we find a star coloring of $G$. 
    We now show that $G$ has a star coloring using at most $k$ colors 
    if and only if there exists a feasible assignment to ILP. 
    
    \begin{lemma}\label{lem:ilp-color}
    If there exists a feasible assignment to ILP then $G$ has a star coloring using at most $k$ colors. 
    \end{lemma}
    \begin{proof}
    Using a feasible assignment returned by  the ILP, we construct a star coloring  $f:V(G)\rightarrow [k]$ of $G$.
    Let $A_1, A_2, \ldots A_{2^t}$ be the subsets of $[t]$ in some fixed order. For each $A_i$, we associate the set of colors $c(A_i)=\{ \sum\limits_{j=0}^{i-1} n_{A_j}+1, \sum\limits_{j=0}^{i-1} n_{A_j}+2, \ldots, \sum\limits_{j=0}^{i-1} n_{A_j}+n_{A_i}\}$, where $n_{A_0}=0$.  
    
    

    Now, for each $V_j$, we associate the set of colors $c(V_j)= \cup_{j \in A_i} c(A_i)$.
    If $V_j$ is a clique type, then from constraint (C2), $|c(V_j)|=|V_j|$ for every $j$. Therefore, we color the vertices of $V_j$ with distinct colors from the set $c(V_j)$. 
If $V_j$ is an independent type, then from constraint  (C3), $1 \leq |c(V_j)| \leq \min \{k, |V_j|\}$. In this case, we greedily color the vertices of $V_j$ with colors from the set $c(V_j)$ such that each color is used at least once in $V_j$. This finishes the description of the coloring $f$ of $G$.  

We now argue that $f$ is a star coloring of $G$.
To show that $f$ is a proper coloring, we need to show that every vertex is assigned a color and adjacent vertices do not receive the same color. 
The coloring process described above ensures that every vertex is colored. 
Also, $f$ is a proper coloring because of the constraints (C0) and (C2). The former constraint ensures that subset types $T_A$ considered do not contain a pair of adjacent types in it while the latter constraint ensures that no two vertices in a clique type are assigned the same color. 
Thus $f$ is a proper coloring. 

We now show that there is no bi-colored path of length $3$. Suppose that there exists a path $u_1-u_2-u_3-u_4$ on four vertices such that $f(u_1)=f(u_3)$ and $f(u_2)=f(u_4)$. 

\begin{itemize}\setlength\itemsep{1em}
    \item \textbf{Vertices $u_1,u_2,u_3, u_4$ belong to four distinct types. }
    
    WLOG, let $u_1, u_2, u_3, u_4$ belong to $V_1, V_2, V_3, V_4$ respectively. 
    From the definition of neighborhood diversity, we have $V_1, V_3 \in adj(V_2)$ and $V_4 \in adj(V_3)$. 
    As $f(u_1)=f(u_3)$ and $f(u_2)=f(u_4)$, there exists two sets $A\subseteq [t]$ and $B\subseteq [t]$ such that 
    $V_1, V_3\in  T_A$, $V_2, V_4\in  T_B$, $n_A \geq 1$ and $n_B\geq 1$. This cannot happen because of the constraint (C4). 

    \item \textbf{Vertices $u_1,u_2,u_3, u_4$ belong to three  distinct types.}   
        \vspace{0.2cm}

      WLOG, let $u_1, u_2, u_3, u_4$ belong to $V_1, V_2, V_1, V_3$ respectively. Since $f(u_1)=f(u_3)$, it is the case that $\sum\limits_{A:V_1\in T_A} n_A < |V_1|$ implying $V_1$ is an independent type. 
      Since $V_1$ is an independent type with two vertices assigned the same color and  $f(u_2)=f(u_4)$, 
      there exists $B\subseteq [t]$ such that $V_2, V_3\in  T_B$ and $n_B\geq 1$. This cannot happen because of the constraint (C5). 
      

    
    \item \textbf{Vertices $u_1,u_2,u_3, u_4$ belong to two   distinct types. } 

     WLOG, let $u_1, u_3\in V_1$ and $ u_2, u_4\in V_2$. Similar arguments as in the above case can be applied to show that $V_1$ and $V_2$ are independent types and this case cannot arise due to constraint (C6). 

\end{itemize}
Thus $f$ is a star coloring of $G$ using at most $k$ colors. 
\qed
\end{proof}
    
 \begin{lemma}\label{lem:color-ilp}
 If $G$ has a star coloring using at most $k$ colors then there exists a feasible assignment to ILP. 
 \end{lemma}  

\begin{proof}
Let $f:V(G)\rightarrow [k]$ be a star coloring of $G$ using $k$ colors. 
For each $A\subseteq [t]$, we set $$n_A=|\bigcap \limits_{V_i \in T_A} f(V_i) - \bigcup \limits_{V_i \notin T_A} f(V_i)|.$$ That is    
$n_A$ is  the number of colors that appear in each of the types in $T_A$ and does not appear in any of the types from $\{V_1,\ldots, V_t \}\setminus T_A$.  
We now show that such an assignment satisfies the constraints (C0)-(C7). 

\begin{enumerate}
    \item Since $f$ is a proper coloring of $G$, no two vertices in two adjacent types are assigned the same color.  Hence the constraint (C0) is satisfied. 
    \item Using the fact that $f$ is a star coloring that uses $k$ colors and from the definition of $n_A$, where each color is counted towards only exactly 
    one variable, we see that the constraint (C1) is satisfied. For each of the remaining variables $n_A$ for which no color is associated with it, we have that $n_A=0$. Hence the constraint (C7) is satisfied. 
    \item When $V_i$ is a clique type, we have that $|f(V_i)|=|V_i|$. 
    The expression $\sum\limits_{A:V_i\in T_A}n_A$ denotes the number of colors used in $V_i$ in the coloring $f$, which equals 
    $|V_i|$. 
     Hence the constraint (C2) is satisfied.

     \item When $V_i$ is an independent type, the number of colors used in $V_i$ is at most the minimum of $k$ and $|V_i|$. In addition, we need at least one color to color the vertices of $V_i$. Hence $1 \leq |f(V_i)| \leq \min\{k, |V_i|\}$. Since $\sum\limits_{A:V_i\in T_A}n_A =|f(V_i)|$, the constraint (C3) is satisfied. 
     
     \item Since $f$ is a star coloring,  there is no bi-colored $P_4$. 
Thus for every combination of four types, say $V_1, V_2, V_3$ and $V_4$, 
if there exists a color assigned to a vertex in $V_1$ and a vertex in $V_3$ with $V_1, V_3\in adj(V_2)$ and $V_4\in adj(V_3)$, then all the vertices in $V_2\cup V_4$ should be assigned distinct colors. That is, there is no $B\subseteq [t]$ for which $V_2, V_4\in T_B$ and $n_B\geq 1$. Hence the constraint  
(C4) is satisfied.

\end{enumerate}
Similarly, we can show that constraints (C5) and (C6) are also satisfied. 
\qed
\end{proof}

The running time of the algorithm depends on 
the time taken to construct an ILP instance and 
obtain a feasible assignment for the ILP using Theorem \ref{thm:ndilp}. 
The former takes polynomial time while the latter takes $O(q^{2.5q+o(q)}n)$ time where $q=2^t$ is the number of variables. 
This completes the proof of Theorem \ref{thm:nd}. 

\section{Twin Cover}
In this section, we show that \scp{} is FPT when parameterized by twin cover. 
Ganian~\cite{ganian2015improving} introduced the notion of  twin-cover which is a generalization of vertex cover. Note that the parameters neighborhood diversity and  twin-cover are not comparable (see Section~3.4 in \cite{ganian2015improving}). We now define the parameter twin-cover and state some of its properties.

\begin{definition}[Twin Cover~\cite{ganian2015improving}]\label{def:tw}
Two vertices $u$ and $v$ of a graph $G$ are said to be  {\it twins}  if $N(u)\setminus \{v\}=N(v)\setminus \{u\}$ and \emph{true twins} if $N[u]=N[v]$.  
A {\it twin-cover} of a graph $G$ is a set $X \subseteq V (G)$ of vertices such that for every edge $uv \in E(G)$ either $u \in X$ or $v \in X$, or
$u$ and $v$ are true twins.
\end{definition}

\begin{remark}
If $X \subseteq V (G)$ is a twin-cover of $G$ then (i) $G - X$ is disjoint union of
cliques, and (ii) for each clique $K$ in $G - X$ and each pair of vertices $u, v$ in $K$,
$N (u) \cap X = N(v) \cap X$.
\end{remark}

\begin{theorem}\label{thm:tw}
\scp{} can be solved in $O(q^{2.5q+o(q)}n)$ time where $q=2^{2^t}$
and $t$ is the size of a twin-cover of the graph. 
\end{theorem}


\medskip
\noindent
\textbf{Overview of the Algorithm:} Given an instance $(G, k,t)$ of \scp{}, and a twin cover $X\subseteq V(G)$ of size $t$ in $G$, the goal is to check if there exists a star coloring of $G$ using at most $k$ colors. The algorithm consists of the following four steps. 
\begin{enumerate}
    \item We guess the coloring $f: X \rightarrow[t']$ of $X$ in a star coloring of $G$ (where $t' \leq t$). Then construct an auxiliary graph $G'$ from $G$ where the neighborhood diversity of $G'$ is bounded by a function of $t$. 

    \item We show that $G$ has a star coloring $g$, using at most $k$ colors, such that $g|_X=f$ if and only if $G'$ has a star coloring $h$, using at most $k$ colors,  such that $h|_X=f$.
    
    \item We construct a graph $\mathcal{B}$, which is a subgraph of $G'$ such that $G'$ has a star coloring $h$, using at most $k$ colors, such that $h|_X=f$ if and only if $\mathcal{B}$ has a proper coloring using at most $k-t'$ colors, where $t'=|f(X)|$.
    
    \item We show that 
    the neighborhood diversity of $\mathcal{B}$ is bounded by a function of $t$. Then we use the FPT algorithm parameterized by neighborhood diversity 
    from \cite{ganian2015improving} to check whether $\mathcal{B}$ has a proper coloring using at most $k-t'$ colors and 
    decide if there exists a star coloring of $G'$ using at most $k$ colors. 
\end{enumerate}

Given a graph $G$, there exists an algorithm to compute a twin-cover of size at most $t$ (if one exists) in $O(1.2738^t+tn)$ time~\cite{ganian2015improving}.  Hence we assume that we are given a twin-cover $X=\{v_1, v_2, \dots, v_t\}\subseteq V(G)$ of size $t$. 

Let $(G,k,t)$ be an instance of \scp{} and $X=\{v_1, v_2, \dots, v_t\}\subseteq V(G)$ be a twin-cover of size $t$ in $G$. 
That is, $G[V\setminus X]$ is a disjoint union of cliques. By the definition of twin cover, 
all vertices in a clique $K$ from $G[V\setminus X]$ has the same neighborhood in $X$. 
Similar to the proof of Theorem \ref{thm:nd}, we define subset types. 
For each $A\subseteq [t]$, 
let $T_A=\{v_i\mid i\in A\}\subseteq X$ denote a subset type of $G$. 
For every subset type $T_A$, 
we denote a \emph{clique type} of $G$ by $K_A=\{K \mid  K \mbox{ is a clique in } G[V\setminus X] \mbox{ and } N(K) \cap X=T_A \}$. 

Step 1 of the algorithm is to initially guess the colors of the vertices in $X$ in a star coloring of $G$. 
Let $f : X \rightarrow [t']$ be such a coloring, where $t' \leq t$. The rest of the proof is to check if $f$ could be extended  to a coloring $g:V(G) \rightarrow [k]$ such that $g |_{X}=f$. 
Let $X_i=f^{-1}(i)\subseteq X$ be the set of vertices from $X$ that are assigned the color $i\in [t']$ in $f$. 
We now construct an auxiliary graph $G'$ from $G$ by repeated application of the Claims \ref{cla:claim1}, \ref{cla:claim2} and the Reduction Rule \ref{red:1}. 

\begin{claim claim}\label{cla:claim1}
Let $K_A$ be a clique type  with $|K_A|\geq 2$ and there exist two vertices in $X_i\cap T_A$  
for some $i\in [t']$. 
Let $G^{\star}$ be the graph obtained from $G$ by adding additional edges between  every pair of non-adjacent vertices in $\bigcup\limits_{K\in K_A}V(K)$. 
Then 
$(G,k,t)$ is a yes-instance of \scp{} if and only if $(G^{\star}, k,t)$ is a yes-instance of \scp{}. 

\end{claim claim}

\begin{proof}
Let $K,K'\in K_A$ be two cliques 
and $u, v\in X_i\cap T_A$ (i.e., $f(u)=f(v)=i$). For the forward direction, 
let $g$ be a star coloring of $(G,k,t)$. Since $g|_X=f$  and $g$ is a star coloring, 
no two vertices in $\bigcup\limits_{K\in K_A}V(K)$ are assigned the same color. Suppose not, there exists two vertices $w,w'\in \bigcup\limits_{K\in K_A}V(K)$ such that $g(w)=g(w')$, then $w-u-w'-v$ is a bi-colored $P_4$. 
Observe that $g$ is also a star coloring of $(G^{\star}, k, t)$. 

For the reverse direction, let $h$ be a star coloring of $(G^{\star}, k, t)$ that uses at most $k$ colors. 
Since $G$ is a subgraph of $G^{\star}$, we have that $h$ is also a star coloring of $(G, k, t)$. 
\qed
\end{proof}

Notice 
that a clique type $K_A$ satisfying the assumptions of Claim \ref{cla:claim1} will now have 
$|K_A|=1$. 
We now look at the clique types $K_A$ such that $|K_A|\geq 2$ and apply the following reduction rule. 
Let $K\in K_A$ be an arbitrarily chosen  clique with maximum number of vertices.

\begin{reduction rule}\label{red:1}
Let $K_A$ be a clique type  with $|K_A|\geq 2$ and 
$|X_i\cap T_A|\leq 1$, for all $i\in [t']$. Also, let $K\in K_A$ be an arbitrarily chosen  clique with maximum number of vertices over all cliques in $K_A$. 
Then $(G, k, t)$ is a yes-instance of \scp{} if and only if $(G-\bigcup\limits_{K'\in K_A \setminus \{K\} } V(K'), k, t)$ is a yes-instance of \scp{}. 
\end{reduction rule}

\begin{lemma}\label{lem:rrs}
Reduction Rule \ref{red:1} is safe. 
\end{lemma}

\begin{proof} 
Suppose $(G, k,t)$ is a yes-instance of \scp{}. Then it is easy to see that $(G-V(K'), k, t)$ is a yes instance of \scp{}. 
For the reverse direction, let $g$ be a star coloring of $(G-V(K'), k, t)$. 
We show how to extend $g$ to the vertices of $K'$ maintaining the star coloring requirement. 
We use  the colors from $g(K)$ (assigned to the vertices of $K$) to color the vertices of the deleted clique $K'\in K_A$. 
Every vertex in $K'$ is assigned a distinct color from $g(K)$. 
This is possible as $|K'|\leq |K|$. 

We now prove that there is no bi-colored $P_4$ in $G$. 
Suppose not. Let there exist a bi-colored $P_4$ in $G$ because of the coloring assigned to $K'$. Notice that this happens 
only when there exists two vertices in $T_A$ that are assigned the same color. 
In this case, we would have applied our Claim \ref{cla:claim1}, which is a contradiction to the fact that $K_A$ does not satisfy the assumptions of Claim \ref{cla:claim1}. 
\qed
\end{proof}

We repeatedly apply Reduction Rule \ref{red:1} on the clique types $K_A$ for which $|K_A|\geq 2$ after the application of Claim \ref{cla:claim1}. Thereby ensuring  $|K_A|= 1$ for all clique types $K_A$ for which  $|K_A|\geq 2$ in $G$. 
Thus for all clique types $K_A$, we have that $|K_A|\leq 1$. Notice that after the application of Claim~\ref{cla:claim1} and Reduction
Rule~\ref{red:1}, the resulting graph has bounded neighborhood diversity. However, a proper coloring of the resulting graph may not yield a star coloring.  The following claim help us to reduce our problem to proper coloring parameterized by neighborhood diversity. 


\begin{claim claim}\label{cla:claim2}
Let $K_A$ and $K_B$, with $A\neq B$, be two clique types such that there exists two vertices $u,v\in X_i$ such that $u\in T_A\cap T_B$ and $v\in T_B$,  
for some $i\in [t]$. 
Let 
$G^{\star}$ be the graph obtained from $G$ by adding  additional edges between every pair of non-adjacent vertices in 
$V(K_A)\cup V(K_B)$. Then 
$(G,k,t)$ is a yes-instance of \scp{} if and only if $(G^{\star}, k,t)$ is a yes-instance of \scp{}. 
\end{claim claim}
\begin{proof}
For the forward direction, 
let $g$ be a star coloring of $(G,k,t)$. This implies that no two vertices in $V(K_A)\cup V(K_B)$ are assigned the same color because of the vertices $u$ and $v$. Hence $g$ is also a star coloring of $(G^{\star}, k, t)$. 

The reverse direction is trivial. Since $G$ is a subgraph of $G^{\star}$, the star coloring of $(G^{\star}, k, t)$ is also a star coloring of $(G, k, t)$. 
\qed 
\end{proof}
We are now ready to  explain the steps of our algorithm in detail. 

\medskip
\noindent
\textbf{Step 1:} Given an instance $(G, k, t)$ of \scp{}, we construct an auxiliary graph. 
The graph constructed after repeated application  of Claim \ref{cla:claim1}, Reduction Rule \ref{red:1} and Claim \ref{cla:claim2}, 
is the auxiliary graph $G'$. 
We now argue that the neighborhood diversity of $G'$ is bounded by a function of $t$. 
Consider the partition  $\{V(K_A)~|~A \subseteq [t]\} \cup \{\{v_i\}~|~ v_i \in X\}$ of $V(G')$. 
Notice that each clique type $K_A$ of $G'$ is a clique type (see 
Section \ref{sec:nd} for more details). 
That is, all the vertices in $K_A$ have the same neighborhood in $X$. This is true because initially all vertices in $K_A$ have the same neighborhood (by definition of twin cover) and during the process of adding edges (Claims \ref{cla:claim1} and \ref{cla:claim2}), 
either all the vertices in $K_A$ are made adjacent to all the vertices in a type $K_B$ ($A\neq B$) or none of them are adjacent to any vertex in $K_B$. 
Thus the number of such types is at most  $2^t$. 
Including the vertices of $X$, we have that the neighborhood diversity of $G'$ is at most $2^t +t$. 


\medskip\noindent \textbf{Step 2:} We need to show that $(G, k, t)$ is a yes-instance of \scp{} if and only if $(G', k, t)$ is a yes-instance of \scp{}. This is accomplished by the correctness of the Claims \ref{cla:claim1}, \ref{cla:claim2} and the Reduction rule \ref{red:1}.

\medskip\noindent \textbf{Step 3:} 
The next step of the algorithm is to find a set of colors from $[t']$ that can be assigned to the vertices in $V\setminus X$. 
Towards this, 
for each $A \subseteq[t]$,  we guess a subset of colors $D_A \subseteq [t']$  of size at most $|V(K_A)|$, 
that can be assigned to the vertices in the clique type $K_A$ in a star coloring of $G$ (extending the coloring $f$ of $X$)  that uses at most $k$ colors.  
For the guess $D_A$, we arbitrarily (as it does not matter which vertices are assigned a specific color) assign colors from $D_A$ to vertices in $K_A$ such that $|D_A|$ vertices in $K_A$ are colored distinctly. Given the guess $D_A$ for each $K_A$, 
we can check in $2^{O(t)}$ time 
if the color set $D_A$ associated with $K_A$ is indeed a proper coloring (considering the coloring $f$ of $X$ and its  neighboring types). 
In a valid guess,  
some vertices  $Q\subseteq V(G')\setminus X$  are assigned colors from $[t']$. 
The uncolored vertices of $G'$ should be given a color from $[k]\setminus [t']$. 
Let $g:X\cup Q\rightarrow [t']$ be a coloring such that $g(v)=f(v)$ if $v \in X$, and 
$g(v)=\ell$, where $\ell$ is the assigned color as per the above greedy assignment, if $v \in Q$

We now extend this partial coloring $g$ of $Q\cup X$ 
to a full coloring of $G'$, 
where the vertices in $V(G')\setminus (Q\cup X)$ are assigned colors from $[k]\setminus [t']$. 
Let $\mathcal{B}$ be the subgraph of $G'$ obtained by deleting the vertices $Q \cup X$ from $G'$. 
Notice that $\mathcal{B}$ has neighborhood diversity at most $2^t$. 
\begin{claim claim}\label{cla:subgraph-b}
There exists a star coloring of $G'$ extending $g$ 
using at most $k$ colors 
 if and only if there exists a proper coloring of $\mathcal{B}$ using at most $k-t'$ colors. 
\end{claim claim}

\begin{proof}
Let $h: V(G') \rightarrow [k]$ be a star coloring of $G'$ such that $h|_{Q\cup X}=g$. Clearly $|h(\mathcal{B})|=|h(V \setminus (Q\cup X))| \leq k-t'$. That is, $h$ restricted to the vertices of $\mathcal{B}$ is a proper coloring of $\mathcal{B}$ which uses at most $k-t'$ colors. 

For the reverse direction, let $c:V(\mathcal{B})\rightarrow [k-t']$ be a proper coloring. We construct a coloring $h:V(G')\rightarrow [k]$ using the coloring $c$ as follows: 
$h(v)=c(v)$ if $v \in \mathcal{B}$, and  $h(v)=g(v)$ otherwise.
We show that $h$ is a star coloring of $G'$. 
Suppose not, without loss of generality, let $u_1-u_2-u_3-u_4$ be a bi-colored 
$P_4$, with $u_1\in K_{A_1}$, $u_2\in X$, $u_3\in K_{A_2}$ and $u_4\in X$ (notice that this is the only way a bi-colored $P_4$ exists), for some $A_1, A_2\subseteq [t]$. 
That is, $c(u_1)=c(u_3)$ and $c(u_2)=c(u_4)$. Also, $A_1\neq A_2$ because of the proper coloring. If this were the case, 
we would have applied Claim \ref{cla:claim2} as 
$K_{A_1}$ and $K_{A_2}$ satisfy the assumptions along with coloring of the vertices $u_2$ and $u_4$ in $X$. 
As a consequence, each vertex in $K_{A_1}$ would have been adjacent to each vertex in $K_{A_2}$. 
\qed 
\end{proof}



%
\medskip\noindent \textbf{Step 4:} 
It is known that proper coloring is FPT parameterized by neighborhood diversity \cite{ndrobert}. The algorithm in \cite{ndrobert} uses integer linear programming with $2^{2k}$ variables, where $k$ is the neighborhood diversity of the graph. 
Since $\mathcal{B}$ has neighborhood diversity at most $2^t$, we have that the number of variables $q\leq 2^{2^t}$. 
We use the algorithm to test whether 
$\mathcal{B}$ has a proper coloring using at most $k-t'$ colors. 


\medskip 
\noindent{}\textbf{Running time:} Step 1 of the algorithm takes $O(t^t)$ time to guess a coloring of $X$. 
Reduction Rule \ref{red:1}, Claims \ref{cla:claim1} and \ref{cla:claim2} can be applied in $2^{O(t)} n^{O(1)}$ time.  
Step 2 can be processed in $2^{O(t)} n^{O(1)}$ time. 
Step 3 involves guessing the colors that the clique types can take from the colors used in $X$ and this takes $O(2^{2^t})$ time. 
Constructing $\mathcal{B}$ takes polynomial time. 
Step 4 is applying the FPT algorithm parameterized by neighborhood diversity from \cite{ndrobert} on $\mathcal{B}$ which takes  $O(q^{2.5q+o(q)} n)$ time where $q\leq 2^{2^t}$. 
The latter dominates the running time and hence the running time of the algorithm is 
$O(2^{2^t}q^{2.5q+o(q)}n^{O(1)})$, where $q\leq 2^{2^t}$.

 This completes the proof of Theorem \ref{thm:tw}. 
 
\section{Clique-width}\label{sec:cw}

In this section, we show that \scp{} is FPT when parameterized by combined parameter
clique-width and the number of colors. 
We first give the definition of clique-width.



\begin{definition}[Clique-width \cite{courcellecw}]
	Let $w \in \mathbb{N}$.
	A $w$-expression $\Phi$ defines a graph $G_\Phi$ where each vertex of $G$ receives a label from the set $[w]$,
	using the following four recursive operations with indices $i,j \in [w]$, $i\neq j$:

	\begin{enumerate}
	\item  Introduce, $\Phi=v(i)$: 
    $G_{\Phi}$ is a graph consisting a single vertex $v$ with label $i$. 
	\item Disjoint union,  $\Phi=\Phi' \oplus \Phi''$:   $G_{\Phi}$ is a disjoint union of the graphs $G_{\Phi'}$ and $G_{\Phi''}$. 
		
	\item  Relabel, {\bf $\Phi= \rho_{i\rightarrow j}(\Phi')$}:
	$G_\Phi$ is the graph $G_{\Phi'}$ where each vertex labeled $i$ in $G_{\Phi'}$ now has label $j$.
	\item Join, $\Phi=\eta_{i,j}(\Phi')$:
	$G_\Phi$ is the graph $G_{\Phi'}$ with additional edges between each pair of vertices $u$
	    of label $i$ and $v$ of label $j$.
	\end{enumerate}
The \emph{clique-width} of a graph $G$ denoted by \emph{cw(G)} is the minimum integer $w$
    such that there is a $w$-expression $\Phi$ that defines $G$. 
\end{definition}

Given a graph $G=(V,E)$ and an integer $k$, there is an FPT-algorithm that either reports $\operatorname{cw}(G) > w$ or outputs a $(2^{3w+2}-1)$-expression of $G$ \cite{DBLP:journals/jct/OumS06}. 
Hence, we assume that a $w$-expression $\Psi$ of $G$ is given. A $w$-expression $\Psi$ is a \emph{nice} $w$-expression of $G$, 
if no edge is introduced twice in $\Psi$. 
Given a $w$-expression of $G$, it is possible to get a nice $w$-expression of $G$ in polynomial time~\cite{courcellecw}.
For more details on clique-width we refer the reader to~\cite{hlinveny2008width}. 

The main result of this section is the following. 

\begin{theorem}\label{thm:cw}
Given a graph $G$, its nice $w$-expression and an integer $k$, we can decide if there exists a star coloring of $G$ using $k$ colors in $O((3^{w^3k^2+w^2k^2})^2 n^{O(1)})$ time. 
\end{theorem}

\begin{proof}
Let $\Psi$ be a nice $w$-expression of the graph $G$. 
We give a dynamic programming algorithm over $\Psi$. 
 For each subexpression $\Phi$ of $\Psi$ and a coloring  $C:V(G_{\Phi}) \rightarrow [k]$ of $G_{\Phi}$,  we have a boolean table entry $d[\Phi;N;A;B]$ where 
$$N=n_{1,1},\dots n_{1,k},\dots, n_{w,1}, \dots n_{w,k}\mbox{, } $$
$$A=A_1, A_2, \dots, A_w,  \mbox{ where for each } i\in [w]  \mbox{ and } q,q'\in [k],$$
$$  A_i=\{A_{i,\{j,\ell\}}^{q,q'} \mid j,\ell\in [w]\mbox{ and }  j\neq \ell \}\cup \{A_{i, \{j,j\}}^{q, q'}
\mid j\in [w]\} \mbox{, and }$$
$$\mbox{for each } q,q'\in [k], B=\{B_{j, \ell}^{q,q'}\mid 
j, \ell\in [w] \}.
$$

Given some vertex coloring of $G_\Phi$, we explain the meaning of each of the variables below. 

\medskip
\noindent
\textbf{N:} For each label $i\in [w]$ and color $q\in [k]$, the variable $n_{i,q}\in \{0,1,2\}$. Let $n_{i,q}^*$ denote the number of number of vertices with label $i$ and color $q$. Then $n_{i,q}=\max\{2, n_{i,q}^*\}$. The number of variables in $N$ is $wk$. 

\medskip
\noindent
\textbf{A:} Let $L=\{Y\subseteq [w] : |Y|=2\}$. 
That is $|L|=\binom{w}{2}$. 
For each label $i$ and set $\{j,\ell\}\in L$  or set $\{j, j\}\in L$ where $j, \ell\in [w]$,   the variable $A_{i,\{j,\ell\}}^{q,q'}\in \{0, 1, 2\}$. 
 Let $\widehat{A}_{i,\{j,\ell\}}^{q,q'}$ denote the number of vertices with label $i$ and color $q$ such that 
there exists two neighbors assigned the color $q'$, one with label $j$ and the other with label $\ell$. Then 
$A_{i,\{j,\ell\}}^{q,q'}= \max\{2, \widehat{A}_{i,\{j,\ell\}}^{q,q'}\}$. 
 Notice that the number of variables in $A$ is $k^2w(\binom{w}{2}+w)$. 

\medskip 
\noindent
\textbf{B:} 
For each pair of labels $i,j\in [w]$ and pair of colors $q, q'\in [k]$, the variable $B_{i,j}^{q,q'}\in \{0, 1, 2\}$. 
 Let $\widehat{B}_{i,j}^{q,q'}$ denote the number of vertices with label $i$ assigned the color $q$ such that 
there exists a neighbor  with label $j$ assigned the color $q'$. 
Then $B_{i,j}^{q,q'}= \max\{2, \widehat{B}_{i,j}^{q,q'}\}$. The number of variables in $B$ is $(wk)^2$. 

Note that for each $i, j, \ell\in [w]$ and colors $q, q'\in [k]$, 
the variable $A_{i,\{j,\ell\}}^{q,q'}$ (resp. $B_{i, j}^{q,q'}$) corresponds to the number of vertices, limited to a maximum of 2,  with label $i$ (resp. label $i$) adjacent to vertices with labels $j$ and $\ell$ (resp. label $j$) that are assigned the color $q'$. The primary difference between the variables is that the variables in $B$ are defined for each pair of labels $i,j\in [w]$ while the variables in $A$ are defined for each combination of $i$ and label set $\{j,\ell\}\in L$. 

 
 For each subexpression $\Phi$ of $\Psi$, a boolean entry $d[\Phi;N;A;B]$ is set to TRUE if and only if 
 there is a vertex coloring $C:V(G_\Phi)\rightarrow[k]$ that satisfies the variables 
 $n_{i,q}$, $A_{i,\{j,\ell\}}^{q,q'}$ and $B_{i,j}^{q,q'}$ for each $i, j, \ell\in [w]$ and colors $q, q'\in [k]$. 
 If there is no coloring satisfying the variables entries in   $d[\Phi;N;A;B]$, then we set then the entry to FALSE.
We say that $G$ has a star coloring using $k$ colors if and only if there exists an entry $d[\Psi;N;A;B]$ that is set to TRUE. 

We now give the details on how to compute an entry $d[\Phi;N;A;B]$ at each operation. 
\begin{enumerate}\setlength\itemsep{1.2em}
  \item $\Phi=v(i)$. 
  
  $G_{\Phi}$ is the graph with one vertex $v$ with label $i$. The vertex $v$ could take any of the colors from $[k]$, and hence we set $k$ entries to be TRUE, one for each color. 
  That is, the entry $d[\Phi; N;A;B]$ is set to  TRUE if and only if 
  $n_{i, q}=1$ for each color $q\in [k]$ and all other variables of $N$, $A$ and $B$ are 0.

   \item $\Phi= \Phi' \oplus \Phi''$
   
       $G_{\Phi}$ is a  disjoint union of the graphs $G_{\Phi'}$ and $G_{\Phi''}$. 
       We set the entry $d[\Phi;N;A;B]$ to TRUE 
    if and only if there exist entries $d[\Phi';N';A';B']$
    and $d[\Phi''; N'';A'';B'']$
    such that both the entries are TRUE 
    and the following conditions hold:  
    \begin{enumerate}[label=(\roman*)]
        \item For each $i\in [w]$ and $q\in [k]$, 
        $n_{i,q}=\min\{2, n'_{i,q}+n''_{i,q} \}$.

    \item For each $i\in [w]$, set $\{j,\ell\}\in L$  and $q, q'\in [k]$, 
    $A_{i,\{j,\ell\}}^{q,q'}=\min\{2, A_{i,\{j,\ell\}}^{'q,q'}+A_{i,\{j,\ell\}}^{''q,q'}\}$. 
    
       \item For each $i, j\in [w]$,  and $q, q'\in [k]$, 
    $A_{i,\{j,j\}}^{q,q'}=\min\{2, A_{i,\{j,j\}}^{'q,q'}+A_{i,\{j,j\}}^{''q,q'}\}$.

    \item For each $i, j\in [w]$ 
    and $q, q'\in [k]$, 
    $B_{i,j}^{q,q'}=\min\{2, B_{i,j}^{'q,q'}+B_{i,j}^{''q,q'}\}$.


    \end{enumerate}

   We now explain how to set an entry $d[\Phi;N;A;B]$ to TRUE. 
   Initially, we set all the entries in $G_\Phi$ to FALSE. 
   We consider all possible combination of entries $d[\Phi';N';A';B']$ and $d[\Phi'';N'';A'';B'']$ that are TRUE. 
   That is, for each pair of entries $d[\Phi';N';A';B']=$TRUE and $d[\Phi'';N'';A'';B'']$=TRUE, we update the corresponding entry $d[\Phi;N;A;B]$ satisfying the above conditions. 
    The number of such combinations is  $O((3^{w^3k^2+w^2k^2})^2)$ and we can compute the values of $N,A$ and $B$ in $O(w^3k^2)$ time. 
   
  \item  $\Phi=\rho_{i \rightarrow j} (\Phi')$.
  
	$G_{\Phi}$ is obtained from the graph $G_{\Phi'}$ by relabelling the vertices of label $i$ in $G_{\Phi'}$ with label $j$ where $i,j\in [w]$. 
	We set an entry $d[\Phi;N;A;B]$ to TRUE  if and only if there exists an entry $d[\Phi';N';A';B']$ in $G_{\Phi'}$ such that 
	$d[\Phi';N';A';B']$ is TRUE and the following conditions hold: 
	\begin{enumerate}[label=(\roman*)]
	    \item For each color $q\in [k]$ and
	    label $\ell\in [w]\setminus \{i,j\}$, 
			$n_{\ell,q}= n'_{\ell,q}$. 
	\item 	For each color $q\in [k]$,  
			$n_{j,q}= \min\{2, n'_{j,q}+n'_{i,q}\}$ and $n_{i,q}=0$.

			\item For each pair of colors $q, q'\in [k]$, the variables in $B$ are calculated as follows. 
			

	\noindent\fbox{
		\begin{minipage}{0.8\textwidth}
			\begin{tabular*}{\textwidth}{@{\extracolsep{\fill}}lr} \textbf{Computation of variables in $B$}  \\ \end{tabular*}
			\begin{itemize}
        			
	    \item For each pair of labels $a, b\in [w]\setminus \{i,j\}$,  
			$B_{a,b}^{q, q'}= B_{a,b}^{'q,q'}$.


	\item 	For 
	each label $a\in [w]\setminus \{i,j\}$,  
			$B_{a,j}^{q,q'}= \min\{2, B_{a,j}^{'q,q'}+B_{a,i}^{'q,q'}\}$ and 
			$B_{j,a}^{q,q'}= \min\{2, B_{j,a}^{'q,q'}+B_{i,a}^{'q,q'}\}$.

	\item $B_{j,j}^{q,q'}= \min\{2, B_{j,j}^{'q,q'}+B_{i,i}^{'q,q'}+B_{i,j}^{'q,q'}+B_{j,i}^{'q,q'}\}$. 
			
			\item 	For each label $a\in [w]$,  $B_{a, i}^{q,q'}= 0$ and $B_{i,a}^{q,q'}= 0$. 
    \end{itemize}
		\end{minipage}
	}

		\item For each pair of colors $q, q'\in [k]$, the variables in A are calculated as follows. 	
			
			

	\noindent\fbox{
		\begin{minipage}{0.8\textwidth}
			\begin{tabular*}{\textwidth}{@{\extracolsep{\fill}}lr} \textbf{Computation of variables in $A$}  \\ \end{tabular*}
		 \begin{itemize}

		\item For each $\ell, a, b\in [w]\setminus \{i,j\}$ and the set $\{a,b\}\in L$, 	
	    $A_{\ell,\{a,b\}}^{q,q'}=A_{\ell,\{a,b\}}^{'q,q'}$. 
	
	\item For each $\ell, a\in [w]\setminus \{i,j\}$, 
	$A_{\ell,\{a,a\}}^{q,q'}=A_{\ell,\{a,a\}}^{'q,q'}$. 				
						
		\item 	For   
		each $\ell, a\in [w]\setminus \{i, j\}$, 
	$A_{\ell,\{j,a\}}^{q,q'}=\min\{2, A_{\ell,\{i,a\}}^{'q,q'}+A_{\ell,\{j,a\}}^{'q,q'}\}$.

	\item 	For each $\ell \in [w]\setminus \{i, j\}$, 
	$A_{\ell,\{j,j\}}^{q,q'}=\min\{2, A_{\ell,\{j,j\}}^{'q,q'}+A_{\ell,\{i,j\}}^{'q,q'}+A_{\ell,\{i,i\}}^{'q,q'}\}$.

		\item 	For each $a, b\in [w]\setminus \{i, j\}$  and $\{a, b\}\in L$,  
	$A_{j,\{a,b\}}^{q,q'}=\min\{2, A_{j,\{a,b\}}^{'q,q'}+A_{i,\{a,b\}}^{'q,q'}\}$.

		\item 	For each $a\in [w]\setminus \{i, j\}$,  
	$A_{j,\{a,a\}}^{q,q'}=\min\{2, A_{j,\{a,a\}}^{'q,q'}+A_{i,\{a,a\}}^{'q,q'}\}$.

		\item For each $a\in [w]\setminus \{i, j\}$, 
		$A_{j,\{a,j\}}^{q,q'}=\min\{2, A_{j,\{a,j\}}^{'q,q'}+A_{i,\{a,j\}}^{'q,q'}+A_{j,\{a,i\}}^{'q,q'}+A_{i,\{a,i\}}^{'q,q'}\}$.

		\item $A_{j,\{j,j\}}^{q,q'}=\min\{2, A_{j,\{j,j\}}^{'q,q'}+A_{i,\{j,j\}}^{'q,q'}+A_{j,\{i,i\}}^{'q,q'}+A_{i,\{i,i\}}^{'q,q'}+A_{j,\{j,i\}}^{'q,q'}+A_{i,\{j,i\}}^{'q,q'}\}$. 
				
		\item For each $a, b, \ell\in [w]$, 
		$A_{\ell, \{a,b\}}^{q,q'}=0$ if  $a=i$ or $b=i$ or $\ell=i$. 
		
		\end{itemize}
		\end{minipage}
	}

	\end{enumerate}

	 	 We initially set each entry $d[\Phi;N;A;B]$ to be FALSE for each combination of $N, A$ and $B$. We consider all possible entries $d[\Phi';N';A';B']$ such that $d[\Phi';N';A';B']$ is TRUE and set the corresponding entry $d[\Phi;N;A;B]$ to TRUE based on the values computed using the above rules.  The number of entries to check is $O(3^{w^3k^2+w^2k^2})$ and we can compute the values of $N,A$ and $B$ in $O(w^3k^2)$ time.

 \item $\Phi=\eta_{i,j} (\Phi')$
 
$G_\Phi$ is obtained by connecting each vertex of label $i$ with each vertex of label $j$ in $G_{\Phi'}$. 
To ensure a proper coloring, we consider the entries 
$d[\Phi';N';A';B']$ that are set to TRUE in $G_{\Phi'}$ and has the property that, for each $q\in [k]$, 
if $n'_{i,q}\geq 1$ then $n'_{j,q}=0$ and vice-versa. 
This condition ensures that the coloring obtained after the $\eta_{i,j}(\Phi')$ operation is indeed a proper coloring. 
It may be the case that $n'_{i,q}=0$ and $n'_{j,q}=0$, which implies that there are no vertices with labels $i$ and $j$ with the same color $q$. 

Before we proceed to the conditions on how to set an entry $d[\Phi;N;A;B]$ to TRUE, we look at each entry $d[\Phi';N';A';B']$ in $G_{\Phi'}$ that is set to TRUE and check if any of
the following four cases are met. The four cases are illustrated in Figure \ref{figure:star}.

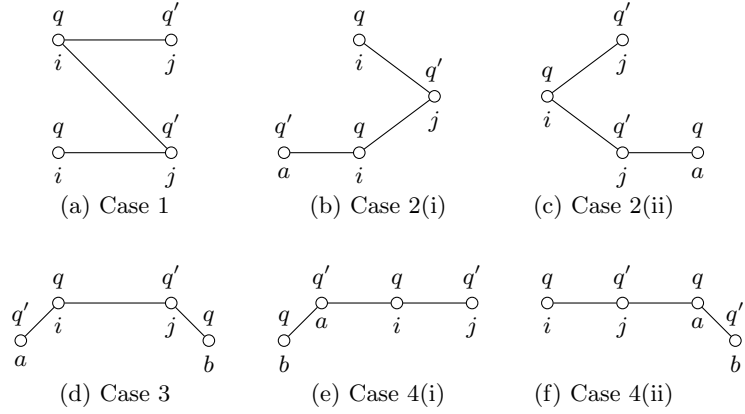
\begin{figure}
\vspace{-0.3cm}
\begin{center}
\begin{tikzpicture}
[scale=0.5,auto=left, node/.style={circle,fill=white, draw, scale=0.5}
	,max/.style={circle,fill=black, draw, scale=1}]

	\node[node] (a1) at (-1,1) {};
	\node[node] (a2) at (-1,-2) {};
	\node[node] (a3) at (-4,1) {};
	\node[node] (a4) at (-4, -2) {};

	\node [above] at (a1.north) {$q'$};
	\node [above] at (a2.north) {$q'$};
	\node [above] at (a3.north) {$q$};
	\node [above] at (a4.north) {$q$};	
	\node [below] at (a1.south) {$j$};
	\node [below] at (a2.south) {$j$};
	\node [below] at (a3.south) {$i$};
	\node [below] at (a4.south) {$i$};

	\node[node] (b1) at (4,1) {};
	\node[node] (b2) at (4,-2) {};
	\node[node] (b3) at (6,-0.5) {};
	\node[node] (b4) at (2, -2) {};

	\node [above] at (b1.north) {$q$};
	\node [above] at (b2.north) {$q$};
	\node [above] at (b3.north) {$q'$};
	\node [above] at (b4.north) {$q'$};	
	\node [below] at (b1.south) {$i$};
	\node [below] at (b2.south) {$i$};
	\node [below] at (b3.south) {$j$};
	\node [below] at (b4.south) {$a$};

		\node[node] (c1) at (11,1) {};
	\node[node] (c2) at (11,-2) {};
	\node[node] (c3) at (9,-0.5) {};
	\node[node] (c4) at (13, -2) {};

	\node [above] at (c1.north) {$q'$};
	\node [above] at (c2.north) {$q'$};
	\node [above] at (c3.north) {$q$};
	\node [above] at (c4.north) {$q$};	
	\node [below] at (c1.south) {$j$};
	\node [below] at (c2.south) {$j$};
	\node [below] at (c3.south) {$i$};
	\node [below] at (c4.south) {$a$};

	\node[node] (d1) at (-1,-6) {};
	\node[node] (d2) at (0,-7) {};
	\node[node] (d3) at (-4,-6) {};
	\node[node] (d4) at (-5, -7) {};

	\node [above] at (d1.north) {$q'$};
	\node [above] at (d2.north) {$q$};
	\node [above] at (d3.north) {$q$};
	\node [above] at (d4.north) {$q'$};	
	\node [below] at (d1.south) {$j$};
	\node [below] at (d2.south) {$b$};
	\node [below] at (d3.south) {$i$};
	\node [below] at (d4.south) {$a$};

		\node[node] (e1) at (5,-6) {};
	\node[node] (e2) at (7,-6) {};
	\node[node] (e3) at (3,-6) {};
	\node[node] (e4) at (2, -7) {};

	\node [above] at (e1.north) {$q$};
	\node [above] at (e2.north) {$q'$};
	\node [above] at (e3.north) {$q'$};
	\node [above] at (e4.north) {$q$};	
	\node [below] at (e1.south) {$i$};
	\node [below] at (e2.south) {$j$};
	\node [below] at (e3.south) {$a$};
	\node [below] at (e4.south) {$b$};

		\node[node] (f1) at (13,-6) {};
	\node[node] (f2) at (14,-7) {};
	\node[node] (f3) at (11,-6) {};
	\node[node] (f4) at (9, -6) {};

	\node [above] at (f1.north) {$q$};
	\node [above] at (f2.north) {$q'$};
	\node [above] at (f3.north) {$q'$};
	\node [above] at (f4.north) {$q$};	
	\node [below] at (f1.south) {$a$};
	\node [below] at (f2.south) {$b$};
	\node [below] at (f3.south) {$j$};
	\node [below] at (f4.south) {$i$};

	\node [above] at (-2.5,-4) {(a) Case 1};
	\node [above] at (4.5,-4) {(b) Case 2(i)};
	\node [above] at (10.5,-4) {(c) Case 2(ii)};
	\node [above] at (-2.5,-9) {(d) Case 3};
	\node [above] at (4.5,-9) {(e) Case 4(i)};
	\node [above] at (10.5,-9) {(f) Case 4(ii)};

	\draw (a1) -- (a3) -- (a2) -- (a4); 
	\draw (b1) -- (b3) -- (b2) -- (b4); 
	\draw (c1) -- (c3) -- (c2) -- (c4); 
	\draw (d2) -- (d1) -- (d3) -- (d4); 
	\draw (e2) -- (e1) -- (e3) -- (e4); 
    \draw (f2) -- (f1) -- (f3) -- (f4);

\end{tikzpicture}
\end{center}
\vspace{-0.3cm}
\caption{Various cases of bicolored $P_4$ that may arise out of the operation $\Phi=\eta_{i,j}(\Phi')$. Each vertex is represented by $\circ$, where its label and color are represented by the values below and above ``$\circ$'' respectively.  }
\label{figure:star}
\vspace{-0.5cm}
\end{figure}

 \begin{itemize}

     \item \textbf{Case 1:} There exists a pair of colors $q, q'\in [k]$ such that $n'_{i, q}=2$ and $n'_{j,q'}=2$.

     \item \textbf{Case 2:} There exists a pair of colors $q, q'\in [k]$ and label $a\in [w]\setminus \{i,j\}$ 
     such that either (i) 
     $n'_{i, q}= 2$,  $n'_{j,q'}= 1$ and $B_{i,a}^{'q,q'}\geq 1$,  
     or 
     (ii) $n'_{i, q}= 1$,  $n'_{j,q'}= 2$ and $B_{j,a}^{'q',q}\geq 1$.

    \item \textbf{Case 3:} There exists a pair of colors 
    $q, q'\in [k]$ and labels $a, b\in [w]\setminus \{i, j\}$ 
    such that $n_{i,q}=1$, $n_{j,q'}=1$, 
     $B_{i,a}^{'q, q'}\geq 1$ and $B_{j,b}^{'q',q}\geq 1$. 

     \item \textbf{Case 4:} There exists a pair of colors 
    $q, q'\in [k]$ and  labels $a, b\in [w]\setminus \{i,j\}$  
    such that $n_{i,q}=1$, $n_{j,q'}=1$, 
     and either $A_{a, \{i,b\}}^{'q', q}\geq 1$ 
     or 
     $A_{a, \{j,b\}}^{'q, q'}\geq 1$. 

\end{itemize}
 
We consider only those entries $d[\Phi';N';A';B']$ that are set to TRUE and do not satisfy any of the above four cases. 
If an entry satisfies any of the above four cases, then the $\eta_{i,j}(\Phi')$ results in a bi-colored $P_4$ 
and hence it is 
not considered for further processing. 
If none of the above cases are met, then we consider the entry $d[\Phi';N';A';B']$ that is set to TRUE  for further processing. 

We set an entry $d[\Phi;N;A;B]$ to be TRUE  if and only if there exists an entry $d[\Phi';N';A';B']$ 
in $G_{\Phi'}$ such that 
	$d[\Phi';N';A';B']$ is set to TRUE, not satisfying any of the above four cases 
	and the following conditions are met: 

    \begin{enumerate}[label=(\roman*)]
        \item  For each $q\in [k]$ and
	  $a\in [w]$, 
			$n_{a,q}= n'_{a,q}$. 
			
	    \item  For each $q, q'\in [k]$ and 
	    $a \in [w]\setminus \{i,j\}$, $b\in [w]$, 
			$B_{a,b}^{q,q'}= B_{a, b}^{'q,q'}$.

			 \item  For each $q, q'\in [k]$ and
	    $a \in [w]\setminus \{j\}$, $b\in [w]\setminus \{i\}$, 
	    $B_{i, a}^{q,q'}= B_{i, a}^{'q,q'}$ and $B_{j, b}^{q,q'}= B_{j, b}^{'q,q'}$.  
	    
	    \item For  each $q, q'\in [k]$, we have  $B_{i,j}^{q,q'}=\min\{2,n'_{i,q}\}$ if $n'_{j,q'}\geq 1$ 
	    and $B_{i,j}^{q,q'}=0$ otherwise. 
	    Similarly, we have 
	    $B_{j,i}^{q,q'}=\min\{2, n'_{j,q}\}$ if $n'_{i,q'}\geq 1$ and 
	    $B_{j,i}^{q,q'}=0$ otherwise. Note that $B_{i,j}^{'q,q'}=0$ and $B_{j,i}^{'q,q'}=0$ in $G_{\Phi'}$ because $\Phi'$ is a nice $w$-expression of $G_{\Phi'}$. 
	    
	   		\item For each pair of colors $q, q'\in [k]$, we compute the variables in $A$ as follows:

	\noindent\fbox{
		\begin{minipage}{0.8\textwidth}
			\begin{tabular*}{\textwidth}{@{\extracolsep{\fill}}lr} \textbf{Computation of variables in $B$}  \\ \end{tabular*}
		 \begin{itemize}
        \item For each label set  $\{a, b\}\in L$ and $\ell\in [w]\setminus \{i,j\}$, we have 
	    $A_{\ell, \{a, b\}}^{q,q'}=A_{\ell, \{a, b\}}^{'q, q'}$.
	    	    
	  \item For each $a\in [w]$ and $\ell\in [w]\setminus \{i,j\}$, 
	   $A_{\ell, \{a, a\}}^{q,q'}=A_{\ell, \{a, a\}}^{'q, q'}$.

	  \item For each label set $\{a, b\}\in L$ and $a, b\in [w]\setminus \{j\}$, 
	   $A_{i,\{a,b\}}^{q,q'}=A_{i,\{a,b\}}^{'q,q,'}$. Also for each $a\in [w]\setminus \{j\}$, $A_{i,\{a,a\}}^{q,q'}=A_{i,\{a,a\}}^{'q,q,'}$.

	 \item For each label set $\{a, b\}\in L$ and $a, b\in [w]\setminus \{i\}$, 
	   $A_{j,\{a,b\}}^{q,q'}=A_{j,\{a,b\}}^{'q,q,'}$.   
	   Also for each $a\in [w]\setminus \{i\}$, $A_{j,\{a,a\}}^{q,q'}=A_{j,\{a,a\}}^{'q,q,'}$.
	   
	 \item For each label $a \in [w]\setminus \{j\}$, we have $A_{i,\{j,a\}}^{q,q'}=0$ if $n'_{j,q'}=0$. 
	 Else we have 
	    $A_{i,\{j,a\}}^{q,q'}=B_{i,a}^{'q,q'}$. 
	    
	    \item For each label $a \in [w]\setminus \{i\}$, we have $A_{j,\{i,a\}}^{q,q'}=0$ if $n'_{i,q'}=0$. 
	 Else we have 
	    $A_{j,\{i,a\}}^{q,q'}=B_{j,a}^{'q,q'}$. 


	    \item $A_{i,\{j,j\}}^{q,q'}=n'_{i,q}$ if $n'_{j,q'}=2$ and $A_{i,\{j,j\}}^{q,q'}=0$ otherwise. 
	    \item 
	    $A_{j,\{i,i\}}^{q,q'}=n'_{j,q}$ if $n'_{i,q'}=2$ and $A_{j,\{i,i\}}^{q,q'}=0$ otherwise. 
			
   \end{itemize}
		\end{minipage}
	}

    \end{enumerate}
    			
		We initially set all the entries in $G_\Phi$ to be FALSE. For each TRUE entry $d[\Phi';N';A';B']$, we check if all the above conditions are met (besides not falling into any of the four cases) and then assign the respective entry $d[\Phi;N;A;B]$ to be TRUE. 
	    The number of entries to check is $O(3^{w^3k^2+w^2k^2})$ and we can compute the values of $N,A$ and $B$ in $O(w^3k^2)$ time.

\end{enumerate}	   
    The correctness of the algorithm follows from the description of the algorithm. The time taken by the algorithm is $O((3^{w^3k^2+w^2k^2})^2n^{O(1)})$. 

\qed
\end{proof}

\section{Conclusion}
In this paper, we study the parameterized complexity of \scp{} with respect to several structural graph parameters. 
We show that \scp{} is FPT when parameterized by (a) neighborhood diversity, (b) twin cover, and (c) the combined parameter clique-width and the number of colors. 

We conclude the paper with the following open problems for further research. 
\begin{enumerate}
    \item What is the parameterized complexity of \scp{} when parameterized by distance to cluster or distance to co-cluster? 
    \item It is known that graph coloring  admits a polynomial kernel when parameterized by distance to clique~\cite{gutin2021parameterized}. Does \scp{} also admit a polynomial kernel parameterized by distance to clique?
\end{enumerate}

\medskip
\noindent
\textbf{Acknowledgments:} We would like to thank anonymous referees for their helpful comments. 
The first author and the second author acknowledges  SERB-DST 
for supporting this research via grants PDF/2021/003452 and SRG/2020/001162 respectively for funding to support this research. 

\bibliography{BibFile}

\end{document}